\documentclass[12pt,a4paper]{article}
\usepackage{epsf, graphicx}
\usepackage{latexsym,amsfonts,amsbsy,amssymb}
\usepackage{amsmath,amsthm}
\usepackage{booktabs}
\usepackage{blkarray}
\usepackage{tikz}
\usepackage{color,xcolor}
\usepackage{cleveref}
\usepackage{indentfirst}
\usepackage{bm}
\usepackage [latin1]{inputenc}
\makeatletter

\@addtoreset{equation}{section} \makeatother
\newtheorem{Theorem}{Theorem}[section]
\newtheorem{Lemma}{Lemma}[section]
\newtheorem{Corollary}{Corollary}[section]
\newtheorem{Definition}{Definition}[section]
\newtheorem{Proposition}{Proposition}[section]
\theoremstyle{definition}
\newtheorem{Example}{Example}[section]
\newtheorem{Remark}{Remark}[section]
\makeatletter \@addtoreset{equation}{section} \makeatother
\newcommand{\FF}{\mathbb{F}}
\newcommand{\Cc}{\mathcal{C}}

\begin{document}
\title{Constructions of non-Generalized Reed-Solomon MDS codes }

\author{Shengwei Liu,~Hongwei Liu,~Fr\'{e}d\'{e}rique Oggier}
\date{}
\maketitle
\insert\footins{\small
\noindent Shengwei Liu and Hongwei Liu are with the School of Mathematics and Statistics and the Key Laboratory NAA-MOE, Central China Normal University, Wuhan 430079,
China (e-mail: shengweiliu@mails.ccnu.edu.cn, hwliu@ccnu.edu.cn).\\
Fr\'{e}d\'{e}rique Oggier is currently with the School of Mathematics, the University of Birmingham, UK, B15 2TT
 (e-mail:
f.e.oggier@bham.ac.uk). Part of this work was done at Nanyang Technological University, Singapore.\\
 }
{\centering\section*{Abstract}}
 \addcontentsline{toc}{section}{\protect Abstract} 
 \setcounter{equation}{0} 

 Generalized Reed-Solomon codes form the most prominent class of maximum distance separable (MDS) codes, codes that are optimal in the sense that their minimum distance cannot be improved for a given length and code size. The study of codes that are MDS yet not generalized Reed-Solomon codes, called non-generalized Reed-Solomon MDS codes, started with the work by Roth and Lemple (1989), where the first examples where exhibited. It then gained traction thanks to the work by Beelen (2017), who introduced twisted Reed-Solomon codes, and showed that families of such codes are non-generalized Reed-Solomon MDS codes. Finding non-generalized Reed-Solomon MDS codes is naturally motivated by the classification of MDS codes.


In this paper, we provide a generic construction of MDS codes, yielding infinitely many examples. We then explicit families of non-generalized Reed-Solomon MDS codes. Finally we position some of the proposed codes with respect to generalized twisted Reed-Solomon codes, and provide new view points on this family of codes.

\medskip
\noindent{\large\bf Keywords:~}\medskip MDS codes, non-GRS MDS codes

\noindent{\bf2010 Mathematics Subject Classification}: 94B05, 94B65.

\section{Introduction}

For $q$ a prime power, let $\mathbb{F}_{q}$ be the finite field of $q$ elements and $\mathbb{F}_{q}^{*}=\mathbb{F}_{q}\backslash\{0\}$ be its multiplicative group. A linear $[n, k, d]$ code $\Cc$ over $\mathbb{F}_{q}$  is a $k$-dimensional subspace of $\FF_q^n$ (so $k\leq n$) with minimum Hamming distance $d$, meaning that $d$ counts the minimum number of non-zero coordinates over all non-zero vectors, called codewords, in $\Cc$. Basis vectors of $\Cc$ are stacked as rows of a $k\times n$ matrix called a generator matrix of $\Cc$, and a generator matrix  of the form $(I_{k}|A_{k\times(n-k)})$ is said to be a systematic generator matrix, or in systematic form. The dual $\mathcal{C}^\perp$ of $\mathcal{C}$ is the subspace formed by vectors in $\FF_q^n$ which are orthogonal to all codewords of $\mathcal{C}$.

If the parameters $n,k,d$ of the code $\Cc$ reach the Singleton bound, meaning that $d=n-k+1$, then $\Cc$ is called a maximal distance separable (MDS) code. MDS codes are often considered optimal, in that they offer the best trade-off between $d$, representing the erasure/erasure recovery ability of the code, and $n,k$, characterizing the efficiency of the code (how much redundancy for a given amount $k$ of information). They are highly sought after. We recall a classical characterization of MDS codes.

\begin{Proposition}\label{pro1}\cite[p. 319]{R. Roth}
Let $\mathcal{C}$ be a linear $[n,k]$ code over $\mathbb{F}_{q}$. Then $\mathcal{C}$ is MDS if and only if every $k$ columns of a generator matrix of $\mathcal{C}$ are linearly independent.
\end{Proposition}

It is known that the dual of an MDS code is MDS. The most prominent MDS codes are the generalized Reed-Solomon (GRS) codes \cite{R. Roth}.

\begin{Definition}\label{def-1} \cite[p. 323]{R. Roth}
Let $\alpha_{1},\dots,\alpha_{n}\in\mathbb{F}_{q}\bigcup\{\infty\}$ be distinct elements and $v_{1},\dots,v_{n}\in\mathbb{F}_{q}^{*}$. For $k<n$, $\bm v= (v_1,\ldots,v_n)$ and $\bm \alpha= (\alpha_1,\ldots,\alpha_n)$, the corresponding generalized Reed-Solomon (GRS) code is defined by
$$
GRS_{n,k}(\bm\alpha,\bm v) = \{(v_{1}f(\alpha_{1}),\dots,v_{n}f(\alpha_{n})): f\in \mathbb{F}_{q}[x], \deg f<k\}.
$$
In this setting, for a polynomial $f(x)$ of degree $\deg f(x)<k$, the quantity $f(\infty)$ is defined as the coefficient of $x^{k-1}$ in the polynomial $f(x)$. In the case $\bm v= (v_1,\ldots,v_n) =\bm 1=(1,1,\dots,1)$, the code is called a Reed-Solomon (RS) code.
\end{Definition}

GRS codes are very flexible in terms of their parameters, but for the constraint that $n\leq q$ if $\infty$ is not used, and $n\leq q+1$ otherwise. Since they exist for large choices of ${\bm \alpha}, {\bm v}$, $k,n$ and $\FF_q$, it is natural to ask:
\begin{itemize}
\item[(a)]
Given an MDS code, whether this code is a GRS code, i.e., whether there exist some vectors $\bm\alpha$ and $\bm v$ such that the MDS code is of the form $GRS_{n,k}(\bm\alpha,\bm v)$.
\item[(b)] How to construct MDS codes which are not GRS codes.
\end{itemize}

In \cite{G. Seroussi,A. Lempel1}, it was shown that an MDS code with a systematic generator matrix $(I_k|A_{k\times (n-k)})$ is a GRS code if and only if $A$ is an extended (generalized) Cauchy matrix,
a powerful method to determine whether an MDS code is a GRS code, used e.g. in \cite{A. Lempel} to prove that extending a column of the generator matrix of a GRS code creates a type of non-GRS MDS code.

An alternative method relies on the Schur product (see e.g.~\cite{I. Cascudo1},~\cite{H. Randriambololona}).

\begin{Definition}\label{def-schur}
Let $\bm x=(x_{1},\dots,x_{n})$, $\bm y=(y_{1},\dots,y_{n})\in\mathbb{F}_{q}^{n}$, the {\it Schur product} of $\bm x$ and $\bm y$ is defined as $\bm x*\bm y=(x_{1}y_{1},\dots,x_{n}y_{n})$. The product of two linear codes $\mathcal{C}_{1}$, $\mathcal{C}_{2}\subseteq\mathbb{F}_{q}^{n}$ is defined as
$$\mathcal{C}_{1}*\mathcal{C}_{2}=\left\langle\bm c_{1}*\bm c_{2}\mid\bm c_{1}\in\mathcal{C}_{1},~\bm c_{2}\in\mathcal{C}_{2}\right\rangle_{\mathbb{F}_{q}}$$ where $\left\langle S\right\rangle_{\mathbb{F}_{q}}$ denotes the $\mathbb{F}_{q}$-linear subspace generated by the subset $S$ of $\mathbb{F}_{q}^{n}$.

In particular, if $\mathcal{C}_{1}=\mathcal{C}_{2}$, we call $\mathcal{C}^{2}=\mathcal{C}*\mathcal{C}$ the square code of $\mathcal{C}$.
\end{Definition}


\begin{Proposition}\cite{D. Mirandola}\label{pro-zy}
Let $\mathcal{C}\subseteq\mathbb{F}_{q}^{n}$ be an MDS code with $\dim(\mathcal{C})\leq\frac{n-1}{2}$. The code $\mathcal{C}$ is a GRS code if and only if $\dim(\mathcal{C}^{2})=2\dim(\mathcal{C})-1$.
\end{Proposition}

\begin{table}
\centering
\begin{tabular}{l|p{5cm}|p{4.5cm}}
\hline
{\bf Method} &  {\bf Base code} & {\bf References} \\
\hline
Cauchy & extended GRS codes  & \cite{A. Lempel,G. Seroussi,A. Lempel1}\\
  \hline
Schur & GTRS codes \newline algebraic geometry codes  & \cite{P. Beelen1,C.ZHU,C.ZHU2,J.SUI,HYN,YQ,C.ZHU1}  \newline \cite{chenhao} \\
& constacyclic codes & \cite{S.LIU} \\
  \hline
\end{tabular}
\caption{ \label{tab:my_label1}
   A summary of known constructions of non-GRS MDS codes: the column method contains the proof method used to show that the code is not a GRS code (either Cauchy matrices or the Schur product), the base code column gives the code that was modified to obtain a non-GRS MDS code.}
\end{table}

The criterion given by the above proposition was used in \cite{S.LIU}, to provide a sufficient/necessary and sufficient condition for a constacyclic code of length $n$/prime length to be a GRS code, respectively. It was also used in \cite{chenhao}, to construct non-GRS MDS codes from arbitrary genus algebraic curves. Finally, this criterion is also widely used (e.g. in \cite{P. Beelen1,C.ZHU,C.ZHU2,J.SUI,HYN,YQ,C.ZHU1}) in the context of generalized twisted Reed-Solomon (GTRS) codes, a class of codes which generalizes GRS codes (a formal definition is postponed to later, in Definition \ref{def-2}, with its connection to twisted Reed-Solomon codes). We provide a summary of the known construction of non-GRS MDS codes in Table \ref{tab:my_label1}.

\begin{table}
\begin{center}
\begin{tikzpicture}[baseline=(current bounding box.center), scale=0.7, every node/.style={scale=0.8}]
\draw[color=blue!60](-3,0)--(1,0);
\filldraw[color=blue!60, fill=blue!5, very thick, opacity=0.5](-1,0) circle (2);
\filldraw[color=gray!60, fill=gray!5, very thick,opacity=0.5](1.5,0) circle (2);
\node[] at (1.5,-2.5){\textcolor{gray}{GTRS codes}};
\node[] at (-1,-2.5){\textcolor{blue}{MDS codes}};
\node[] at (-1,-1){\textcolor{blue}{GRS codes}};
\node[] at (-1,1){\textcolor{blue}{non-GRS codes}};
\end{tikzpicture}
\hspace{5mm}
\begin{tabular}{c|cc}
 & GRS & non-GRS \\
\hline
MDS GTRS & $\checkmark$ & $\checkmark$ \\
non-GTRS & & \\
\end{tabular}
\caption{\label{fig:gtrs}
Interactions between generalized twisted Reed-Solomon (GTRS)  codes and other codes: on the left,
GTRS codes contain non-MDS codes, GRS (thus MDS) codes and non-GRS MDS codes. On the right, we zoom into the four known classes of MDS codes. Checkmarks highlight the contribution of this paper.}
\end{center}
\end{table}

Through the lens of GTRS codes, as highlighted in Table \ref{fig:gtrs}, the question of constructing MDS codes which are or not GRS codes further split into determining whether these MDS codes are, or not generalized twisted Reed-Solomon codes.
The end motivation between these distinctions is some classification of MDS codes. To start addressing it, it is needed to have large families of MDS codes, which, for now, can be classified into these four classes (see Table \ref{fig:gtrs}), before being possibly further classified into other/finer classes.

Our contributions fit into the above picture. First, in Section 2, we provide a generic construction of MDS codes, yielding arbitrary many examples of MDS codes, out of which we explicit several families of non-GRS codes. Then in Section 3, we look at MDS GTRS codes.
We give an example of a GRS code which is also a GTRS code, implying that GTRS codes also contain GRS codes. We also prove that some the exhibited non-GRS codes from Section 2 are GTRS.
To the best of our knowledge, there is no standard method to prove that an MDS code is non-GTRS, therefore deciding whether examples of non-GTRS codes are found within the proposed construction is left for future work. Nevertheless, we prove some result that provides some alternative view point on (non-)GTRS codes.

\section{A Construction of non-GRS MDS Codes}

Let $\alpha_{1},\dots,\alpha_{k}\in\mathbb{F}_{q}\bigcup\{\infty\}$ be distinct elements and $V_{\bm\alpha}$ be the Vandermonde matrix associated with $\alpha_{1},\dots,\alpha_{k}$, i.e., if $\infty\notin\{\alpha_{1},\dots,\alpha_{k}\}$ then
$$
V_{\bm\alpha}=\left(
\begin{matrix}
1&1&\dots&1\\
\alpha_{1}&\alpha_{2}&\dots&\alpha_{k}\\
\vdots&\vdots&&\vdots\\
\alpha_{1}^{k-1}&\alpha_{2}^{k-1}&\dots&\alpha_{k}^{k-1}
\end{matrix}
\right)_{k\times k},
$$
if instead $\infty=\alpha_{k}$ then
$$
V_{\bm\alpha}=\left(
\begin{matrix}
1&1&\dots&1&0\\
\alpha_{1}&\alpha_{2}&\dots&\alpha_{k-1}&0\\
\vdots&\vdots&&\vdots&\vdots\\
\alpha_{1}^{k-1}&\alpha_{2}^{k-1}&\dots&\alpha_{k-1}^{k-1}&1
\end{matrix}
\right)_{k\times k}.
$$

Let $E_{ij}$ denote the matrix whose $(i,j)$th entry is $1$ and others are zeros.
Consider the $k\times k$ matrix $X$ whose $(i,j)$th entry is either $0$, or a positive power of the indeterminate $x$ if $(i,j)\in \mathcal{I}$ for $\mathcal{I}$ an arbitrary non-empty choice of positions in a $k\times k$ matrix. Alternatively, we may write
$$
X=\sum_{(i,j)\in\mathcal{I}}x^{s_{ij}}E_{ij}
$$
where $(i,j)\in \mathcal{I}$ and $s_{ij}\geq 1$ for $(i,j)\in \mathcal{I}$.
Then the determinant of the matrix
\begin{equation*}
V_{\bm\alpha}+X
\end{equation*}
is by construction a polynomial in  $\mathbb{F}_{q}[x]$.

Throughout this paper, $\bm v$ denotes a vector over $\mathbb{F}_{q}$ and all entries of $\bm v$ are nonzero.
\begin{Proposition}\label{pro31}
With the notations above, $\det(V_{\bm\alpha}+X)$ is a non-zero polynomial in $\FF_q[x]$. The same holds for $V_{\bm \alpha}D_{\bm v}$ for $D_{\bm v}$ a diagonal matrix with the coefficients of ${\bf v}$ on its diagonal.



\end{Proposition}
\begin{proof}
Write
\[
V_{\bm\alpha}+X 
=
\begin{pmatrix}
{\bf v}_1 + {\bf x}_1 &  {\bf v}_2 + {\bf x}_2 &
\ldots &
{\bf v}_k + {\bf x}_k
\end{pmatrix},
\]
for ${\bf v}_1,\ldots,{\bf v}_k$ columns of $V_{\bm\alpha}$
and ${\bf x}_1,\ldots,{\bf x}_k$ columns of $X$. 

By multilinearity of the determinant,
\[
\det(V_{\bm\alpha}+X)
=
\sum_{i=1}^{2^{k}}\det(Y^{(i)})
\]
where each $Y^{(i)}$ is a matrix whose $j$th column is either ${\bf v}_{j}$ or ${\bf x}_{j}$. Since the first column is either ${\bf v}_{1}$ or ${\bf x}_{1}$ (2 choices), then independently, the second column is either ${\bf v}_{2}$ or ${\bf x}_{2}$ (2 choices, for a total of 4 choices), until the $k$ columns are filled, there are indeed $2^k$ such matrices.

Say we label by $Y^{(1)}$ the matrix whose columns are all from $V_{\bm\alpha}$.
Let $Y$ be any one of the $Y^{(2)}, \ldots,Y^{(2^k)}$, whose $(i,j)$th entry is $y_{ij}$.
Then $$\det(Y)=\sum_{\pi}{\rm sgn}(\pi)y_{1\pi(1)}y_{2\pi(2)}\dots y_{k\pi(k)}$$
where $\pi$ runs over the permutations of order $k$, sgn($\pi$) denotes its signature, and for each choice of $\pi$, at least one of the $y_{i\pi(i)}$ is an element of $X$ since by choice of $Y$, at least one of its columns belongs to $X$. By definition of $X$,  $y_{i\pi(i)}$ is either $0$ or a positive power of $x$.  Thus for each choice of $\pi$, $y_{1\pi(1)}y_{2\pi(2)}\dots y_{k\pi(k)}$ is either $0$ or a multiple of a positive power of $x$, and $\det(Y)$ is accordingly either $0$ or a polynomial which is a multiple of $x$. This shows that
$$
\det(V_{\bm\alpha}+X) = \det(V_{\bm\alpha})+ xh(x)
$$
and so is a non-zero polynomial in $x$ since $\det(V_{\bm\alpha})\neq 0$.


\end{proof}

By the definition of GRS codes (see Definition \ref{def-1}), if $\infty\notin\{\alpha_{1},\dots,\alpha_{n}\}$ then a generator matrix of $GRS_{n,k}(\bm\alpha,\bm v)$ is
\begin{equation*}
G_{\bm\alpha,\bm v}=\left(
\begin{matrix}
v_{1}&v_{2}&\dots&v_{n}\\
v_{1}\alpha_{1}&v_{2}\alpha_{2}&\dots&v_{n}\alpha_{n}\\
v_{1}\alpha_{1}^{2}&v_{2}\alpha_{2}^{2}&\dots&v_{n}\alpha_{n}^{2}\\
\vdots&\vdots&&\vdots\\
v_{1}\alpha_{1}^{k-1}&v_{2}\alpha_{2}^{k-1}&\dots&v_{n}\alpha_{n}^{k-1}
\end{matrix}
\right)_{k\times n},~n\leq q,
\end{equation*}
if instead $\infty=\alpha_{n}$ then
a codeword is of the form $(v_{1}f(\alpha_{1}),\dots,v_{n-1}f(\alpha_{n-1}), v_n)$ and thus a generator matrix of $GRS_{n,k}(\bm\alpha,\bm v)$ is
\begin{equation}\label{eq:Galphav}
G_{\bm\alpha,\bm v}=\left(
\begin{matrix}
v_{1}&v_{2}&\dots&v_{n-1}&0\\
v_{1}\alpha_{1}&v_{2}\alpha_{2}&\dots&v_{n-1}\alpha_{n-1}&0\\
v_{1}\alpha_{1}^{2}&v_{2}\alpha_{2}^{2}&\dots&v_{n-1}\alpha_{n-1}^{2}&0\\
\vdots&\vdots&&\vdots&\vdots\\
v_{1}\alpha_{1}^{k-1}&v_{2}\alpha_{2}^{k-1}&\dots&v_{n-1}\alpha_{n-1}^{k-1}&v_{n}
\end{matrix}
\right)_{k\times n},~n\leq q+1.
\end{equation}

Let $\mathbb{F}_{q^{b}}/\mathbb{F}_{q}$ be a field extension, $b>1$ and $\beta\in\mathbb{F}_{q^{b}}\setminus\mathbb{F}_{q}$.
From $G_{\bm\alpha,\bm v}$ defined in (\ref{eq:Galphav}) over $\mathbb{F}_{q}$, build the new generator matrix
\[
G_{\bm\alpha,\bm v} +
\sum_{(i,j)\in\mathcal{J}}\beta^{s_{ij}}E_{ij}
\]
over $\mathbb{F}_{q^{b}}$ for $\mathcal{J}$ an arbitrary non-empty choice of positions in a $k\times n$ matrix.

\begin{Theorem}\label{th31}
With the notations above, let $\mathcal{C}$ be the linear code over $\mathbb{F}_{q^{b}}$ generated by $G_{\bm \alpha, \bm v} + \sum_{(i,j)\in\mathcal{J}}\beta^{s_{ij}}E_{ij}$.
If the degree of the polynomial obtained by picking any $k$ columns of $G_{\bm \alpha,\bm v} + \sum_{(i,j)\in\mathcal{J}}x^{s_{ij}}E_{ij}\in\FF_q[x]$ and computing the determinant of the resulting matrix is strictly less than the degree of the minimal polynomial of $\beta$ over $\mathbb{F}_{q}$, then $\mathcal{C}$ is an MDS code.
\end{Theorem}
\begin{proof}
To prove that $\mathcal{C}$ is an MDS code,
by Proposition \ref{pro1}, we need to prove that any $k$ columns of $G_{\bm\alpha,\bm v} +
\sum_{(i,j)\in\mathcal{J}}\beta^{s_{ij}}E_{ij}$ are linearly independent.

The determinant of any $k$ columns of $G_{\bm\alpha,\bm v} + \sum_{\mathcal{J}} \beta^{s_{ij}} E_{ij}$ is a polynomial in $\beta$, say $g(\beta)$, which is the evaluation in $\beta$ of the polynomial $g(x)$  obtained by picking any $k$ columns of $G_{\bm\alpha,\bm v} +  \sum_{(i,j)\in\mathcal{J}}x^{s_{ij}}E_{ij}$ and computing its determinant.

By Proposition \ref{pro31}, any such $g(x)$ is a non-zero polynomial, therefore $g(\beta)$ cannot be $0$ since its degree is strictly less than the degree of the minimal polynomial of $\beta$ over $\mathbb{F}_{q}$.
\end{proof}

\begin{Corollary}
The method of the above theorem provides infinitely many examples of MDS codes.
\end{Corollary}
\begin{proof}
For a given $k$, we compute all the $k\times k$ submatrices of $G_{\bm \alpha,\bm v} + \sum_{(i,j)\in\mathcal{J}}x^{s_{ij}}E_{ij}$, each will give a polynomial in $\FF_q[x]$ and we identify the highest polynomial degree $d$. Once this highest degree $d$ is known, it is enough to choose $b$ large enough to find a $\beta$ whose minimal polynomial has a degree larger than $d$, to ensure the MDS property, thus providing infinitely many examples.
\end{proof}


For an $[n,1,n]$ linear MDS code over $\mathbb{F}_{q}$, a generator matrix is a vector ${\bm v}$ of length $n$ whose entries are all non-zero, which means the code is a repetition code up to a scaling of each coordinate. This vector $(\bm v)$ is also a generator matrix for a GRS code of dimension 1. Its dual code is then an $[n,n-1,2]$ MDS GRS code and conversely, any $[n,n-1,2]$ MDS code is the dual code of a repetition code (again up to a coordinate scaling).

For an $[n,2,n-1]$ MDS code over $\mathbb{F}_{q}$, since the minimum weight of a non-zero codeword is $n-1$, we can choose such a vector to be the first row of a generator matrix, and find a linearly independent vector as a second row, so they generate the code. Up to a non-zero scaling of each column, we may suppose the non-zero coefficients of the first row of the generator matrix are all $1$, and up to permutations of the columns, that the zero coefficient is in the last column. Since
from (\ref{eq:Galphav}), we know that
\[
G_{\bm\alpha,\bm 1}=\left(
\begin{matrix}
 1& 1 &\dots& 1&0\\
\alpha_{1}&\alpha_{2}&\dots&\alpha_{n-1} & 1\\
\end{matrix}
\right)_{2\times n},~n\leq q+1,
\]
to prove that the code is a GRS code,
it is enough to show that for an MDS code, the second row contains $n-1$ distinct coefficients, followed by 1.
By Proposition \ref{pro1}, any $2\times2$ submatrix has a  non-zero determinant. This forces all coefficients to be pairwise distinct and non-zero. This means that up to a scaling of the last column, the non-zero coefficient may be assumed to be 1, as desired. Finally as before, an $[n,n-2,3]$ MDS code is the dual of an MDS code with parameters $[n,2,n-1]$.

The classification of linear MDS codes for $k=1,2,n-1,n-2$ is thus settled and  we assume from here to the end of this section that $3\leq k\leq n-3$ . For $k=3$, explicit counting of GRS codes compared to MDS codes are provided in \cite[Table 1]{BeelenGlynn} for some small values of $q$ and $n$.

Next we explicit several families of non-GRS MDS codes. To prove that they are not GRS codes, we will use the Schur product approach.

\begin{Proposition}\label{prop1}
With the notations above, let $3\leq k\leq \frac{n-1}{2}$ and $\beta\in\mathbb{F}_{q^{b}}$ but $\beta\not\in \mathbb{F}_{q}$. Let $\mathcal{C}$ be  the linear code  over $\mathbb{F}_{q^{b}}$ generated by the matrix $G_{\bm\alpha,\bm v} + \sum_{(i,j)\in\mathcal{J}}\beta E_{ij}$ which we assume has three consecutive rows $\bm g_{j},\bm g_{j+1},\bm g_{j+2}$ such that the vector $\bm e=\bm g_{j}*\bm g_{j+2}-\bm g_{j+1}*\bm g_{j+1}$ contains at least one entry of the form $a\beta$ or $a\beta \pm\beta^2$ for some non-zero $a\in\mathbb{F}_{q}$ and for $\bm g_{i}$ the $i$th row of the generator matrix. If the degree of the minimal polynomial of $\beta$ over $\mathbb{F}_{q}$ is strictly larger than $4k$, then there exists some $\bm v$ such that  $\mathcal{C}$ is a non-GRS MDS code.
\end{Proposition}
\begin{proof}
It is clear that the determinant of any $k\times k$ submatrix of  $G_{\bm\alpha,\bm v} + \sum_{(i,j)\in\mathcal{J}}x E_{ij}$  will result in a polynomial in $\FF_q[x]$ of degree at most $k$. By Theorem \ref{th31},  $\mathcal{C}$ is an MDS code.

Denote by $\bm z_{i}$ the $i$th row of the matrix $\sum_{(i,j)\in\mathcal{J}}\beta E_{ij}$ whose coefficients are either $0$ or $\beta$ by assumption:
$$\sum_{(i,j)\in\mathcal{J}}\beta E_{ij}=\left(
\begin{matrix}
\bm z_{1}\\
\bm z_{2}\\
\vdots\\
\bm z_{k}
\end{matrix}
\right).
$$

Recall that $\bm v=(v_{1},v_{2},\dots,v_{n})$. Set $\bm\alpha^{i}=(\alpha_{1}^{i},\alpha_{2}^{i},\dots,\alpha_{n}^{i})$ where $\infty^{i}=0$ when $0\leq i\leq k-2$ and $\infty^{k-1}=1$ and similarly $\bm v^{2}=(v_{1}^{2},v_{2}^{2},\dots,v_{n}^{2})$. Then the $j$th row of  $G_{\bm\alpha,\bm v} + \sum_{(i,j)\in\mathcal{J}}\beta E_{ij}$ is $\bm v*\bm\alpha^{j-1}+\bm z_{j}$.
Then
\begin{equation*}
\begin{split}
\bm e=&\bm g_{j}*\bm g_{j+2}-\bm g_{j+1}*\bm g_{j+1}\\
=&(\bm v*\bm\alpha^{j-1}+\bm z_{j})*(\bm v*\bm\alpha^{j+1}+\bm z_{j+2})-(\bm v*\bm\alpha^{j}+\bm z_{j+1})*(\bm v*\bm\alpha^{j}+\bm z_{j+1})\\
=&\bm v*\bm\alpha^{j+1}*\bm z_{j}+\bm v*\bm\alpha^{j-1}*\bm z_{j+2}+\bm z_{j}*\bm z_{j+2}-2\bm v*\bm\alpha^{j}*\bm z_{j+1}-\bm z_{j+1}*\bm z_{j+1}\\
=&\bm v*\bm\alpha^{j-1}*(\bm\alpha^{2}*\bm z_{j}+\bm z_{j+2}-2\bm\alpha*\bm z_{j+1})+(\bm z_{j}*\bm z_{j+2}-\bm z_{j+1}*\bm z_{j+1})
\end{split}
\end{equation*}
belongs to $\mathcal{C}^{2}$.

Since each entry of $\bm z_{j},\bm z_{j+1},\bm z_{j+2}$ is either $0$ or $\beta$, the entries of $\bm v*\bm\alpha^{j-1}*(\bm\alpha^{2}*\bm z_{j}+\bm z_{j+2}-2\bm\alpha*\bm z_{j+1})$ are of the form $b\beta$ for some $b\in\mathbb{F}_{q}$ possibly zero. Those of $\bm z_{j}*\bm z_{j+2}-\bm z_{j+1}*\bm z_{j+1}$ are either $0$ or $\pm\beta^{2}$.
Thus the entries of $\bm e$ take value in $\{0,b\beta, \pm\beta^2, b\beta \pm\beta^2\}$
for some $b\in\mathbb{F}_{q}$. And by assumption, there is at least one entry of $\bm e$ is $b\beta$ or $b\beta \pm\beta^2$ for some nonzero $b\in\mathbb{F}_{q}$.

Then $\mathcal{C}^2$ contains $2k$ vectors, written as the $2k$ rows of the following matrix:
$$
S=\left(
\begin{matrix}
v^2_{1}& v^2_{2}&\dots&v^2_{n}\\
v^2_{1}\alpha_{1}&v^2_{2}\alpha_{2}&\dots&v^2_{n}\alpha_{n}\\
 \vdots & \vdots & &\vdots \\
v^2_{1}\alpha_{1}^{2k-2}&v^2_{2}\alpha_{2}^{2k-2}&\dots &v^2_{n}\alpha_{n}^{2k-2}\\
0& 0 & \ldots     & 0
\end{matrix}
\right)+\left(
\begin{matrix}
\bm m_{1}\\
\bm m_{2}\\
\vdots\\
\bm m_{2k-1}\\
\bm e
\end{matrix}
\right)
$$
where the entries of each $\bm m_{j}$ also take value in $\{0,b\beta, \pm\beta^2, b\beta \pm\beta^2\}$ for some $b\in\mathbb{F}_{q}$ and $\infty^{i}=0$ for $0\leq i\leq 2k-3$ and $\infty^{2k-2}=1$.

First, it is clear that the determinant of any $2k\times 2k$ submatrix of  $S$  will result in a polynomial in $\beta$ over $\FF_q$ of degree at most $4k$. By assumption, if there exists at least one $2k\times 2k$ submatrix of $S$ which results in a non-zero polynomial in $\beta$ then  $\mathcal{C}^2$ contains $2k$ linear independent vectors which implies $\mathcal{C}$ is a non-GRS code via Proposition \ref{pro-zy}.

Without loss of generality, we suppose the first $2k$ entries of $\bm e$ contain at least one $b\beta$ or $b\beta \pm\beta^2$ with $b\neq0$.

Similar to the proof of Proposition \ref{pro31}, we calculate the determinant of $S_{2k}$ which is the $2k\times 2k$ submatrix consisting of the first $2k$ columns of $S$ by multilinearity.

Denote $S_{2k}$ by $M_{\bm\alpha,\bm v}+M$ where
$$
M_{\bm\alpha,\bm v}=\left(
\begin{matrix}
v^2_{1}& v^2_{2}&\dots&v^2_{2k}\\
v^2_{1}\alpha_{1}&v^2_{2}\alpha_{2}&\dots&v^2_{2k}\alpha_{2k}\\
 \vdots & \vdots & &\vdots \\
v^2_{1}\alpha_{1}^{2k-2}&v^2_{2}\alpha_{2}^{2k-2}&\dots &v^2_{2k}\alpha_{n}^{2k-2}\\
0& 0 & \ldots     & 0
\end{matrix}
\right)
$$
and
$$
M=\left(
\begin{matrix}
\bm m^{'}_{1}\\
\bm m^{'}_{2}\\
\vdots\\
\bm m^{'}_{2k-1}\\
\bm e^{'}
\end{matrix}
\right)
$$
where $\bm m^{'}_{i}$ is the vector obtained from deleting the last $n-2k$ coordinates of $\bm m_{i}$ for $1\leq i\leq 2k-1$ and $\bm e^{'}$ is the vector obtained from deleting the last $n-2k$ coordinates of $\bm e$.

Write
\[
M_{\bm\alpha,\bm v}+M
=
\begin{pmatrix}
{\bf r}_1 + {\bf t}_1 &  {\bf r}_2 + {\bf t}_2 &
\ldots &
{\bf r_{2k}} + {\bf t_{2k}}
\end{pmatrix},
\]
for ${\bf r}_1,\ldots,{\bf r_{2k}}$ columns of $M_{\bm\alpha,\bm v}$
and ${\bf t}_1,\ldots,{\bf t_{2k}}$ columns of $M$.

By multilinearity of the determinant,
\[
\det(M_{\bm\alpha,\bm v}+M)
=
\sum_{i=1}^{2^{2k}}\det(P^{(i)})
\]
where each $P^{(i)}$ is a matrix whose $j$th column is either ${\bf r}_{j}$ or ${\bf t}_{j}$.

For the $P^{(i)}$ whose columns contain at least $2$ vectors from $M$, their determinants are $\beta^{2}h(\beta)$. And for the $2k$ matrices whose columns contain only one vector from $M$, the sum of their determinant is
$$
\prod_{i}^{2k}v_{i}^{2}(v_{s_{1}}^{-1}a_{1}\beta+v_{s_{2}}^{-1}a_{2}\beta+\dots+v_{s_{\ell}}^{-1}a_{\ell}\beta)+c\beta^{2}
$$
where $s_{i}$ are all the indices of $\bm e^{'}$ whose corresponding entries are $b\beta$ or $d\beta\pm\beta^{2}$ with $bd\neq0$, since $\bm e=\bm v*\bm\alpha^{j-1}*(\bm\alpha^{2}*\bm z_{j}+\bm z_{j+2}-2\bm\alpha*\bm z_{j+1})+(\bm z_{j}*\bm z_{j+2}-\bm z_{j+1}*\bm z_{j+1})$ that implies all nonzero terms $b\beta$ come from $\bm v*\bm\alpha^{j-1}*(\bm\alpha^{2}*\bm z_{j}+\bm z_{j+2}-2\bm\alpha*\bm z_{j+1})$ so that all $a_{i}$ are just related to $\bm\alpha$ and independent to $\bm v$ and all $a_{i}\neq0$ since the $a_{i}$ are determinant of Vandermonde matrices.

If $$
\prod_{i}^{2k}v_{i}^{2}(v_{s_{1}}^{-1}a_{1}\beta+v_{s_{2}}^{-1}a_{2}\beta+\dots+v_{s_{\ell}}^{-1}a_{\ell}\beta)\neq 0
$$
then $\det(M_{\bm\alpha,\bm v}+M)=a\beta+\beta^{2}h^{'}(\beta)\neq0$.

If not, we can choose other $\bm v^{'}$ such that $v^{'}_{s_{1}}\neq v_{s_{1}}$ and others are same. Then for the $\bm v^{'}$, $\det(M_{\bm\alpha,\bm v^{'}}+M)\neq0$ which complete the proof.

\end{proof}
\begin{Corollary}
With the same notations as Proposition \ref{prop1} and its proof, suppose $0,\infty\notin\{\alpha_{1},\alpha_{2},\dots,\alpha_{n}\}$. Then for some fixed $j_{0}$, let $\bm z_{j_{0}}\neq\bm 0$ and $\bm z_{j_{0}+1}=\bm z_{j_{0}+2}=\bm 0$ and $\bm z_{i}$ be arbitrary for any $i\neq j_{0},j_{0}+1,j_{0}+2$, then we get non-GRS MDS codes over a large enough finite field.
\end{Corollary}
\begin{proof}
From the proof of Proposition \ref{prop1}, if $\bm v*\bm\alpha^{j-1}*(\bm\alpha^{2}*\bm z_{j_{0}}+\bm z_{j_{0}+2}-2\bm\alpha*\bm z_{j_{0}+1})\neq \bm 0$ then $\bm e$ satisfies the assumption. Then for a suitable $\bm v$, we can get non-GRS MDS codes over a large enough finite field.
\end{proof}
\begin{Example}
Let $q=11^{13}$, $\mathbb{F}_{11^{13}}=\mathbb{F}_{11}(\theta)$. Let $\bm v=\bm 1$.
When
$\bm\alpha=(1,2,3,4,5,6,7)$ then
$$
\left(
\begin{matrix}
1+\theta&1+\theta&1&1&1+\theta&1&1\\
1&2&3&4&5&6&7\\
1&4&9&5&3&3&5
\end{matrix}
\right)
$$
generates a $[7,3]$ non-GRS MDS code over $\mathbb{F}_{11^{13}}$.
\end{Example}

Next, we also give two families of non-GRS MDS codes with less constraint on the degree of field extension.
\begin{Proposition}\label{pro33}
Let $3\leq k\leq \frac{n-1}{2}$ and $\beta\in\mathbb{F}_{q^{b}}$ but $\beta\not\in \mathbb{F}_{q}$. Let the linear code $\mathcal{C}$ over $\mathbb{F}_{q^{b}}$ be generated by the matrix $G_{\bm\alpha,\bm v} + \sum_{i\in \mathcal{I}}\beta E_{i1}$
for $\mathcal{I}$ containing at least one instance of three consecutive indices,
and $\alpha_{1}\neq0,1,\infty$. Then $\mathcal{C}$ is a non-GRS MDS code.
\end{Proposition}
\begin{proof}
Since $\beta\in\mathbb{F}_{q^{b}}$ but $\beta\not\in\mathbb{F}_{q}$, the degree of the minimal polynomial of $\beta$ over $\mathbb{F}_{q}$ is at least $2$. Since only the first column of $G_{\bm\alpha,\bm v} + \sum_{i\in\mathcal{I}}x E_{i1}$ contains terms in $x$, the determinant of any $k\times k$ submatrix will result in a polynomial in $\FF_q[x]$ of degree at most 1. By Theorem \ref{th31},  $\mathcal{C}$ is an MDS code.

We first suppose $\infty\notin\{\alpha_{1},\dots,\alpha_{n}\}$, then
$$
\left\{
\begin{array}{ll}
(v_{1}\alpha_{1}^{i-1}+\beta,v_{2}\alpha_{2}^{i-1},\dots,v_{n}\alpha_{n}^{i-1}) & \mbox{ if }i\in \mathcal{I}\\
(v_{1}\alpha_{1}^{i-1},v_{2}\alpha_{2}^{i-1},\dots,v_{n}\alpha_{n}^{i-1}) & \mbox{ else}\\
\end{array}
\right.
$$
is the $i$th row of $G_{\alpha,v} + \sum_{i\in\mathcal{I}}\beta E_{i1}$, $1\leq i\leq k$, and we know that $\mathcal{C}^{2}$ is generated over $\mathbb{F}_{q^{b}}$ by the Schur product of these basis vectors, namely
$$
\left\{
\begin{array}{ll}
((v_1\alpha_1^{s-1}+\beta)(v_{1}\alpha_{1}^{i-1}+\beta),v^2_{2}\alpha_{2}^{i+s-2},\dots,v^2_{n}\alpha_{n}^{i+s-2}) & \mbox{ if }i,s\in \mathcal{I}\\
(v_1\alpha_1^{s-1}(v_{1}\alpha_{1}^{i-1}+\beta),v^2_{2}\alpha_{2}^{i+s-2},\dots,v^2_{n}\alpha_{n}^{i+s-2}) & \mbox{ if only }i\in \mathcal{I}\\
(v^2_{1}\alpha_{1}^{i+s-2},v^2_{2}\alpha_{2}^{i+s-2},\dots,v^2_{n}\alpha_{n}^{i+s-2}) & \mbox{else}.
\end{array}
\right.
$$

Then $\mathcal{C}^2$ contains $2k-1$ vectors, written as the first $2k-1$ rows of the following matrix:
$$
\begin{pmatrix}
*& v^2_{2}&\dots&v^2_{n}\\
*&v^2_{2}\alpha_{2}&\dots&v^2_{n}\alpha_{n}\\
 \vdots & \vdots & &\vdots \\
*&v^2_{2}\alpha_{2}^{2k-2}&\dots &v^2_{n}\alpha_{n}^{2k-2}\\
v_1\beta\alpha_1^{l-1}(\alpha_1-1)^2 & 0 & \ldots     & 0
\end{pmatrix}
$$
where $*$ denotes some element in $\mathbb{F}_{q^{b}}$ and while the last row is computed from the first case, using the three consecutive indices $l$, $l+1$, $l+2$ in $\mathcal{I}$: for $i=l$ and $s=l+2$, we have
$$
(v_1^2\alpha_1^{2l}+v_1\beta(\alpha_1^{l+1}+ \alpha_{1}^{l-1})+\beta^2,v^2_{2}\alpha_{2}^{2l},\dots,v^2_{n}\alpha_{n}^{2l})
$$
and for $i=s=l+1$,
$$
(v_1^2\alpha_1^{2l}+2v_1\beta\alpha_1^{l}+\beta^2,v^2_{2}\alpha_{2}^{2l},\dots,v^2_{n}\alpha_{n}^{2l})
$$
so their difference gives the last row.


Using the assumptions on $\alpha_1$, we claim that the above $2k$ vectors are linearly independent since the determinant of the first $2k$ columns of the above matrix
is non-zero. Thus $\dim(\mathcal{C}^{2})\geq 2k$ which implies $\mathcal{C}$ is non-GRS via Proposition \ref{pro-zy}.

For the case $\infty\in\{\alpha_{1},\dots,\alpha_{n}\}$, $\mathcal{C}$ is a non-GRS code since it is an extended code of the above non-GRS MDS code.
\end{proof}

\begin{Example}
Let $q=7^{2}$, $\mathbb{F}_{7^{2}}=\mathbb{F}_{7}(\theta)$ with $\theta^{2}+2=0$. Let $\bm v=\bm 1$.
When
$\bm\alpha=(2,3,4,5,6,1,0)$ then
$$
\left(
\begin{matrix}
1+\theta&1&1&1&1&1&1\\
2+\theta&3&4&5&6&1&0\\
4+\theta&2&2&4&1&1&0
\end{matrix}
\right)
$$
and when
$\bm\alpha=(2,3,4,5,6,1,0,\infty)$
$$
\left(
\begin{matrix}
1+\theta&1&1&1&1&1&1&0\\
2+\theta&3&4&5&6&1&0&0\\
4+\theta&2&2&4&1&1&0&1
\end{matrix}
\right).
$$
Theses matrices
generate a $[7,3]$ and a $[8,3]$ non-GRS MDS code over $\mathbb{F}_{49}$ respectively.
\end{Example}

\begin{Proposition}\label{pro34}
Let $3\leq k\leq \frac{n-1}{2}$ and $\beta\in\mathbb{F}_{q^{b}}$ but $\beta\not\in\mathbb{F}_{q}$. Let the linear code $\mathcal{C}$ over $\mathbb{F}_{q^{b}}$ be generated by the matrix $G_{\bm\alpha,\bm v} + \beta E_{11}$ where $\alpha_{1}\neq0$. Then $\mathcal{C}$ is an non-GRS MDS code.
\end{Proposition}
\begin{proof}
The MDS property is proven in the same way as in the previous proposition. We are left to prove that this code is not a GRS code.

We first suppose $\infty\notin\{\alpha_{1},\dots,\alpha_{n}\}$, let
$$
\left\{
\begin{array}{ll}
{\bf g}_1=(v_{1}+\beta,v_{2},\dots,v_{n}) & i=1\\
{\bf g}_i=(v_{1}\alpha_{1}^{i-1},v_{2}\alpha_{2}^{i-1},\dots,v_{n}\alpha_{n}^{i-1}) & 2\leq i\leq k\\
\end{array}
\right.
$$
be the $i$th row of $G_{\bm\alpha,\bm v} + \beta E_{11}$.

Then $\mathcal{C}^2$ contains $2k$ vectors, written as the rows of the following matrix:
$$
\left(
\begin{matrix}
(v_{1}+\beta)^{2}&v_{2}^{2}&\dots&v_{n}^{2}\\
v_{1}^{2}\alpha_{1}+\beta v_{1}\alpha_{1}&v_{2}^{2}\alpha_{2}&\dots &v_{n}^{2}\alpha_{n}\\
v_{1}^{2}\alpha_{1}^{2}&v_{2}^{2}\alpha_{2}^{2}&\dots&v_{n}^{2}\alpha_{n}^{2}\\
\vdots&\vdots&&\vdots\\
v_{1}^{2}\alpha_{1}^{2k-2}&v_{2}^{2}\alpha_{2}^{2k-2}&\dots&v_{n}^{2}\alpha_{n}^{2k-2}\\
\beta v_{1}\alpha_{1}^{2}&0&\dots&0
\end{matrix}
\right)
$$
where the first row is the Schur product ${\bf g}_1*{\bf g}_1$, the second row is ${\bf g}_{1}*{\bf g}_{2}$, the third to the $(2k-1)$th rows are ${\bf g}_{i}*{\bf g}_{j}$ for $i+j=2,3,\dots,2k-2$ and $2\leq i,j\leq k$, and the last row is ${\bf g}_{1}*{\bf g}_{3}-{\bf g}_{2}*{\bf g}_{2}$.

Using the assumptions on $\alpha_1$, we claim that the above $2k$ vectors are linearly independent since the determinant of the first $2k$ columns of the above matrix is non-zero. Thus $\dim(\mathcal{C}^{2})\geq 2k$ which implies $\mathcal{C}$ is non-GRS via Proposition \ref{pro-zy}.

For the case $\infty\in\{\alpha_{1},\dots,\alpha_{n}\}$ and $\alpha_{1}\neq\infty$, $\mathcal{C}$ is a non-GRS code since it is an extended code of the above non-GRS MDS code.

For the case $\alpha_{1}=\infty$,
$\mathcal{C}^2$ contains $2k$ vectors, written as the rows of the following matrix:
$$
\left(
\begin{matrix}
\beta^{2}&v_{2}&\dots&v_{n}\\
0&v_{2}^{2}\alpha_{2}&\dots &v_{n}^{2}\alpha_{n}\\
0&v_{2}^{2}\alpha_{2}^{2}&\dots&v_{n}^{2}\alpha_{n}^{2}\\
\vdots&\vdots&&\vdots\\
0&v_{2}^{2}\alpha_{2}^{2k-3}&\dots&v_{n}^{2}\alpha_{n}^{2k-3}\\
v_{1}^{2}&v_{2}^{2}\alpha_{2}^{2k-2}&\dots&v_{n}^{2}\alpha_{n}^{2k-2}\\
\beta v_{1}&0&\dots&0
\end{matrix}
\right)
$$
where the first row is computed from ${\bf g}_{1}*{\bf g}_{1}$, the second is ${\bf g}_{1}*{\bf g}_{2}$, the third to the $(2k-1)$th row are ${\bf g}_{i}*{\bf g}_{j}$ for $i+j=2,3,\dots,2k-2$ and $2\leq i,j\leq k$, and the last row is ${\bf g}_{1}*{\bf g}_{k}-{\bf g}_{2}*{\bf g}_{k-1}$.

We claim that the above $2k$ vectors are linearly independent since the determinant of the first $2k$ columns of the above matrix is non-zero, which completes the proof.
\end{proof}
\begin{Remark}
In Propositions \ref{pro33} and \ref{pro34}, we construct $[n,k]$ non-GRS MDS codes over $\mathbb{F}_{q^{b}}$ with $n\leq q+1$ and $3\leq k\leq\frac{n-1}{2}$. Furthermore, in those two propositions, we added the extra entries in the first column of $G_{\bm\alpha,\bm v}$. Actually, the extra entries can be added in any one column  as long as we give the constraint on the corresponding coordinate of $\bm\alpha$.
\end{Remark}

\begin{Example}
Let $q=7^{2}$, $\mathbb{F}_{7^{2}}=\mathbb{F}_{7}(\theta)$ with $\theta^{2}+2=0$. Let $\bm\alpha=(1,\dots,6,0,\infty)$ and $\bm v=\bm 1$, then
$$
\left(
\begin{matrix}
1+\theta&1&1&1&1&1&1&0\\
1&2&3&4&5&6&0&0\\
1&4&2&2&4&1&0&1
\end{matrix}
\right)
$$
and
$$
\left(
\begin{matrix}
1&1&1&1&1&1&1&\theta\\
1&2&3&4&5&6&0&0\\
1&4&2&2&4&1&0&1
\end{matrix}
\right)
$$
generate each a $[8,3]$ non-GRS MDS code over $\mathbb{F}_{49}$.
\end{Example}

\section{GRS and non-GRS MDS GTRS Codes}



Let $k \leq n$ and $\ell\geq 1$ be positive integers. Set the vectors $\bm h=(h_{1},h_{2},\ldots,h_{\ell})$ with $s$ distinct coefficients and all coefficients in $\{0,1,\ldots,k-1\}$, $\bm t=(t_{1},t_{2},\ldots,t_{\ell})$ with coefficients in $\{1,\ldots,n-k\}$ and $\bm \eta=(\eta_{1},\eta_{2},\ldots,\eta_{\ell})\subseteq\mathbb{F}_{q}^l$. Furthermore, the tuples $(h_i,t_i)$ are distinct, so ${\bm t}$, ${\bm h}$ may have repeated entries, just not in the same coordinates.
The set of $[\bm t,\bm h,\bm \eta]$-{\it twisted polynomials} in $x$ is by definition
$$
\mathcal{P}_{k,n,\ell}[\bm t,\bm h,\bm\eta]
=\Big{\{}\sum_{i=0}^{k-1}f_{i}x^{i}+\sum_{j=1}^{\ell}\eta_{j}f_{h_{j}}x^{k-1+t_{j}}
\,|\,
f_{i}\in\mathbb{F}_{q},i=0,\dots,k-1\Big{\}}\subseteq\mathbb{F}_{q}[x].
$$

We observe that a choice of $(h_j,t_j)$ corresponds to one coefficient $\eta_j$. Having $\eta_j=0$ is thus equivalently obtained by removing the corresponding choice of $(h_j,t_j)$ from ${\bm h}\times{\bm t}$, that is reducing $\ell$. As a consequence we may, without loss of generality, assume that ${\bm \eta}\in(\FF_q^*)^\ell$.

\begin{Definition}\cite{P. Beelen1}\label{def-2}
Let $\alpha_{1},\dots,\alpha_{n}\in\mathbb{F}_{q}$ be distinct elements and $v_{1},\dots,v_{n}\in\mathbb{F}_{q}^{*}$. Set $\bm v= (v_1,\ldots,v_n)$, $\bm \alpha= (\alpha_1,\ldots,\alpha_n)$ and consider $\mathbf{t},\mathbf{h},\bm\eta,\ell$ and $\mathcal{P}_{n,k,\ell}[\mathbf{t},\mathbf{h},\bm\eta]$ as defined above. The generalized twisted Reed-Solomon (GTRS) code of length $n$
is defined by
$$
GTRS_{n,k,\ell}[\bm\alpha,\bm t,\bm h,\bm\eta,\bm v]=\{
(v_1f(\alpha_{1}),v_2f(\alpha_{2}),...,v_nf(\alpha_{n}))
\,|\, f\in\mathcal{P}_{k,n,\ell}[\mathbf{t},\mathbf{h},\bm\eta]\}.
$$
\end{Definition}

Twisted Reed-Solomon codes were first introduced in \cite{P. Beelen,P. Beelen1}. In \cite{P. Beelen}, the vector $\bm h$ has $\ell$ distinct entries. In \cite{P. Beelen1}, this assumption was replaced by asking the tuples $(h_i, t_i)$ for $i = 1, \ldots,\ell$ to be distinct. Some authors (see \cite{WEIDONG} for instance) preferred the name twisted generalized Reed-Solomon (TGRS) codes, which emphasizes the relation to generalized Reed-Solomon codes. Indeed, the GTRS codes are GRS codes if $\bm\eta=(0,\dots,0)$. Both the names twisted generalized Reed-Solomon code (TGRS), and generalized twisted Reed-Solomon code (GTRS) are used interchangeably.
Generalized twisted Reed-Solomon codes are considered in a variety of contexts, such as self-dual codes, LCD codes, MDS and near-MDS codes, to name a few (see e.g. \cite{WEIDONG}-\cite{HYNL}, \cite{H.LIU}, \cite{LUO}, \cite{HS}-\cite{C.ZHU2}).
As mentioned in the introduction, the context of interest for this paper is that of (non-)GRS codes. Table \ref{tab:gtrs} summarizes the results that are thus most relevant, namely the design of non-GRS MDS codes, finding necessary and or sufficient for GTRS codes to be  MDS codes, and the construction of non-GRS MDS codes from the constructed MDS GTRS codes.

A polynomial in $\mathcal{P}_{k,n.\ell}[\mathbf{t},\mathbf{h},\bm\eta]$ is of the form
$f(x) = \sum_{i=0}^{k-1}f_{i}x^{i}+\sum_{j=1}^{\ell}\eta_{j}f_{h_{j}}x^{k-1+t_{j}} $ and we denote by
$f(\bm\alpha)=(f(\alpha_{1}),f(\alpha_{2}),\dots,f(\alpha_{n}))$ its evaluation in ${\bm \alpha}$.
A generator matrix for the $GTRS_{n,k,\ell}[\bm\alpha,\bm t,\bm h,\bm\eta,\bm v]$ code is thus obtained by finding $k$ such polynomials $f^{(0)},\ldots,f^{(k-1)}$ whose evaluation vectors are linearly independent:
$$\left(
\begin{matrix}
f^{(0)}(\bm\alpha)\\
f^{(1)}(\bm\alpha)\\
\vdots\\
f^{(k-1)}(\bm\alpha)\\
\end{matrix}
\right)D_{\bm v}$$
where $D_{\bm v}$ is a diagonal matrix with the coefficients of ${\bf v}$ on its diagonal. Choose $f^{(l)}(x)$ such that $f_{l}=1$ and $f_{m}=0$ for $m\neq l$, that is
\[
f^{(l)}(x) = x^l +
\sum_{j,h_j=l}\eta_jx^{k-1+t_j}
\]
so $f^{(l)}(x) = x^l$ if there is no $j$ such that $h_j=l$.

Set $\bm\alpha^{i}=(\alpha_{1}^{i},\alpha_{2}^{i},\dots,\alpha_{n}^{i})$, we can write
$$
\left(
\begin{matrix}
f^{(0)}(\bm\alpha)\\
f^{(1)}(\bm\alpha)\\
\vdots\\
f^{(k-1)}(\bm\alpha)\\
\end{matrix}
\right)=
\left(
\begin{matrix}
\bm 1\\
\bm\alpha\\
\vdots\\
\bm\alpha^{k-1}
\end{matrix}
\right)+
\left(
\begin{matrix}
\bm e_{0}\\
\bm e_{1}\\
\vdots\\
\bm e_{k-1}\\
\end{matrix}
\right),
$$
where
${\bm e}_l = {\bf 0} $ when
$l \not\in \{h_{1},h_{2},\dots,h_{\ell}\}$, and
$
{\bm e}_l = \sum_{j,h_j=l}\eta_j{\bm\alpha}^{k-1+t_j}
$ else.


Since $1\leq t_j\leq n-k$,  $k\leq k-1+t_{j}\leq n-1$ then we know there is a $k\times (n-k)$ matrix $M$ such that
$$
\left(
\begin{matrix}
\bm e_{0}\\
\bm e_{1}\\
\vdots\\
\bm e_{k-1}\\
\end{matrix}
\right)=
M\left(
\begin{matrix}
\bm\alpha^{k}\\
\bm\alpha^{k+1}\\
\vdots\\
\bm\alpha^{n-1}
\end{matrix}
\right).
$$
By construction of ${\bm e}_l$, the matrix $M$ has one non-zero row per distinct $h_j$, for a total of $s$ non-zero row(s). For a given $h_j$, there can be at most $n-k$ values for $t_j$, corresponding to the $n-k$ columns of $M$. If we were allowing $(h_1,t_1)=(h_2,t_2)$, the term $\eta_1+\eta_2$ would be found in $M$. But under the assumption that $(h_i,t_i)$ are distinct for $i=1,\ldots,\ell$, at most one term $\eta_j$ is found for each entry of $M$. Recalling that ${\bm \eta}\in(\FF_q^*)^\ell$, we get a total of at most $k(n-k)$ coefficients $\eta_j$, $j=1,\ldots,\ell$, implying that $\ell \leq k(n-k)$.

Then a generator matrix of an $[n,k]$ GTRS code is given by
\begin{equation}\label{eq:gtrs}
(I_{k}|M)\left(
\begin{matrix}
\bm 1\\
\bm\alpha\\
\vdots\\
\bm\alpha^{n-1}
\end{matrix}
\right)D_{\bm v}.
\end{equation}

Conversely, for any matrix $M=(m_{ij})$ of size $k\times(n-k)$, the matrix
\begin{equation*}
(I_{k}|M)\left(
\begin{matrix}
\bm 1\\
\bm\alpha\\
\vdots\\
\bm\alpha^{n-1}
\end{matrix}
\right)D_{\bm v}
\end{equation*}
generates a GTRS code $GTRS_{n,k,\ell}[\bm\alpha,\bm t,\bm h,\bm\eta,\bm v]$. In this case, $\ell$ is just the number of non-zero entries in $M$, $s$ is the number of non-zero rows of $M$ and each $(h_{i},t_{i})$ decides the position of the  corresponding non-zero entry in $M$.

\begin{Example}\label{ex:gm}
Take $k=2\leq n=4$ and $\ell=3$ (while we assume $k\geq 3$ in the section, we still use $k=2$ for this example for the sake of simplicity). Set ${\bm h}=(0,0,1)$  with $s=2$ distinct coefficients in $\{0,k-1=1\}$, ${\bm t}=(1,2,2)\in \{1,n-k=2\}$ and ${\bm \eta}=(\eta_1,\eta_2,\eta_3)\in \FF_q^3$.  A generic polynomial for these parameters is thus
\begin{eqnarray*}
f(x)
& = & \sum_{i=0}^{k-1}f_{i}x^{i}+\sum_{j=1}^{\ell}\eta_{j}f_{h_{j}}x^{k-1+t_{j}} \\
& = & f_0 + f_1x+
\eta_1f_0x^2 +
\eta_2f_0x^3 +
\eta_3f_1x^3,
\end{eqnarray*}
$f^{(0)}(x)=1+\eta_1x^2 +
\eta_2x^3$, $f^{(1)}(x)=x+\eta_3x^3$
and
\begin{eqnarray*}
\left(
\begin{matrix}
f^{(0)}(\bm\alpha)\\
f^{(1)}(\bm\alpha)
\end{matrix}
\right)
&=&
\begin{pmatrix}
\bm 1\\
\bm\alpha\\
\end{pmatrix}
+
\begin{pmatrix}
\eta_1\alpha_1^2+ \eta_2\alpha_1^3& \eta_1\alpha_2^2+ \eta_2\alpha_2^3& \eta_1\alpha_3^2+ \eta_2\alpha_3^3& \eta_1\alpha_4^2+  \eta_2\alpha_4^3\\
\eta_3\alpha_1^3 & \eta_3\alpha_2^3 & \eta_3\alpha_3^3 & \eta_3\alpha_4^3
\end{pmatrix}
\\
&=&\left(
\begin{matrix}
\bm 1\\
\bm\alpha\\
\end{matrix}
\right)+
\left(
\begin{matrix}
\eta_1 & \eta_2 \\
\eta_3 & 0 \\
\end{matrix}
\right)
\begin{pmatrix}
\bm \alpha^2\\
\bm\alpha^3\\
\end{pmatrix}.
\end{eqnarray*}
\end{Example}

\begin{table}
\begin{tabular}{l|p{6cm}|p{4cm}}
{\bf Question} & {\bf Parameters} & {\bf References} \\
\hline
MDS condition  & $\ell=s=1$, two special $\bm t, \bm h$ &\cite{P. Beelen1} \\
               & $\ell=s=1$, some different $\bm t, \bm h$ & \cite{J.SUI} \\
               & $s=1, \ell=2$  & \cite{C.ZHU2} \\
               & $\ell=s=2$ & \cite{HS}, \cite{YQ} \\
               & $s=2, \ell=4$ &  \cite{sui1} \\
               & $1\leq\ell=s$ & \cite{GU} \\
               & $1\leq s\leq\ell$ &  \cite{DING} \\
               \hline
MDS codes &  $\bm\alpha \subseteq G < \mathbb{F}^{*}_{q}$ and $\eta\in\mathbb{F}^{*}_{q}\setminus G$  & \cite{P. Beelen1}  \\
              &   $1\leq s=\ell$, $\bm\alpha\subset\FF_Q$,  $\mathbb{F}_{q}/\FF_Q$. & \cite{P. Beelen1}\\
              & $s=1, \ell=2$  & \cite{C.ZHU2} \\
              & (*) & \cite{HS}, \cite{YQ}, \cite{sui1}, \cite{GU}  \cite{DING} \\
              \hline
non-GRS MDS codes & (*) & \cite{P. Beelen1}  \\
                 &  (*)  & \cite{J.SUI} \\
                 & $\ell=s=1$ & \cite{P. Beelen1},\cite{C.ZHU1},\cite{C.ZHU} \\
                 &  (*) & \cite{C.ZHU2} \\
\end{tabular}
\caption{\label{tab:gtrs}
Summary of recent results on GTRS codes: (*) means that the parameters are inherited, in the corresponding reference.
}
\end{table}

\begin{Proposition}\label{trs}
Let $\mathcal{C}$ be a linear $[n,k]$ code over $\mathbb{F}_{q}$ with $1\leq k\leq n\leq q$, with generator matrix $G$. If the first $k$ columns of the matrix
$G(V_{\bm\alpha}D_{\bm v})^{-1}$
consist of an invertible matrix where
$$
V_{\bm\alpha}=\left(
\begin{matrix}
1&1&\dots&1\\
\alpha_{1}&\alpha_{2}&\dots&\alpha_{n}\\
\alpha_{1}^{2}&\alpha_{2}^{2}&\dots&\alpha_{n}^{2}\\
\vdots&\vdots&&\vdots\\
\alpha_{1}^{n-1}&\alpha_{2}^{n-1}&\dots&\alpha_{n}^{n-1}
\end{matrix}
\right)
$$
for $\alpha_{1},\alpha_{2},\dots,\alpha_{n}$ $n$ distinct elements in $\mathbb{F}_{q}$,
and $D_{\bm v}$ is a diagonal matrix with the coefficients of ${\bf v}$ on its diagonal,
then $\mathcal{C}$ is a GTRS code.
\end{Proposition}
\begin{proof}
Set $(A_{k\times k}|B_{k\times(n-k)})=G(V_{\bm\alpha}D_{\bm v})^{-1}$. Since
\begin{eqnarray*}
G &= & G(V_{\bm\alpha}D_{\bm v})^{-1}V_{\bm\alpha}D_{\bm v}
=(A|B)V_{\bm\alpha}D_{\bm v},
\end{eqnarray*}
then $A^{-1}G=(I_{k}|A^{-1}B)V_{\bm\alpha}D_{\bm v}$ is also a generator matrix of $\mathcal{C}$. From (\ref{eq:gtrs}), we know that $A^{-1}G$ is a generator matrix of some GTRS code.
\end{proof}

In the following example, we use Proposition \ref{trs} and its proof to show that non-trivial (i.e. $\bm\eta\neq\bm0$) GTRS codes also contain GRS codes which implies not all MDS GTRS codes are non-GRS.

\begin{Example}
Set $q=7$, $n=7$, $\bm\alpha'=(1,\dots,6,0)$ and $\bm v=\bm 1$ so $D_{\bm v}=I_7$, then a generator matrix $G$ of the GRS code $GRS_{7,3}(\bm\alpha',\bm 1)$ is
$$
G =
G_{\bm\alpha'}=\left(
\begin{matrix}
1&1&1&1&1&1&1\\
1&2&3&4&5&6&0\\
1&4&2&2&4&1&0
\end{matrix}
\right).
$$
Set $\bm\alpha=(0,2,\dots,6,1)$ then
$$
G(V_{\bm\alpha}D_{\bm v})^{-1}=
GV_{\bm\alpha}^{-1}=\left(
\begin{matrix}
1&0&0&0&0&0&0\\
1&2&1&1&1&1&0\\
1&1&2&1&1&1&0
\end{matrix}
\right)
=(A|B),~A=\left(
\begin{matrix}
1&0&0\\
1&2&1\\
1&1&2
\end{matrix}
\right)
$$
and
$$
A^{-1}=\left(
\begin{matrix}
1&0&0\\
2&3&2\\
2&2&3
\end{matrix}
\right).
$$
We finally have that
$$
A^{-1}G
=
\begin{pmatrix}
1& 1& 1& 1 &1& 1 &1 \\
0 &2 &1& 4& 4& 1& 2\\
0 &4& 0 &2 &3 &3& 2 \\
\end{pmatrix}
=
\begin{pmatrix}
1&0&0&0&0&0&0\\
0&1&0&5&5&5&0\\
0&0&1&5&5&5&0
\end{pmatrix}V_{\alpha}
$$
is both a generator matrix for the GRS code $GRS_{7,3}(\bm\alpha',\bm 1)$
and for the GTRS code
$GTRS_{7,3}[\bm\alpha,\bm t,\bm h,\bm\eta,\bm v]$ with $\bm v=\bm 1$, $\bm h=(1,1,1,2,2,2)$, $\bm t=(1,2,3,1,2,3)$ and $\bm\eta=(5,5,5,5,5,5)$. The matrix $M$ from (\ref{eq:gtrs}) has the whole zero vector as its first row, since $0\not\in{\bm h}$, and has for 2nd and 3rd row vectors with three non-zero coefficients. The columns of these non-zero coefficients are given by the vector ${\bf t}$, and their values by the vector ${\bm \eta}$.
\end{Example}

A typical generator matrix $G$ for a linear code is in systematic form, that is
$G = (I_k | A_{k\times n-k})$. To compute $G(V_{\bm\alpha}D_{\bm v})^{-1}$,
write the $n\times n$ matrix $(V_{\bm\alpha}D_{\bm v})^{-1}$ as a block matrix:
\[
(V_{\bm\alpha}D_{\bm v})^{-1}
=\begin{pmatrix}
B_{k\times k} & B_{k\times n-k}\\
B_{n-k\times k} & B_{b-k\times n-k}
\end{pmatrix}
\]
so
\begin{eqnarray*}
G(V_{\bm\alpha}D_{\bm v})^{-1}
&=&
(I_k | A_{k\times n-k})
\begin{pmatrix}
B_{k\times k} & B_{k\times n-k}\\
B_{n-k\times k} & B_{b-k\times n-k}
\end{pmatrix}\\
&=&
\begin{pmatrix}
B_{k\times k}+A_{k\times n-k}B_{n-k\times k} |
B_{k\times n-k}+A_{k\times n-k}B_{n-k\times n-k}
\end{pmatrix}.
\end{eqnarray*}
This gives an explicit condition on the $A_{k\times k}$ part of $G$ for the corresponding code to be GTRS.

We provide next an example where the condition of the proposition is not satisfied, and examples where the condition is instead satisfied.

\begin{Example}
Set ${\bm \alpha}$, ${\bm v}={\bf 1}$ and
\[
G =
\begin{pmatrix}
{\bf 1} \\
{\bm \alpha} \\
\vdots \\
{\bm \alpha}^{k-2} \\
{\bm \alpha}^{k} \\
\end{pmatrix}.
\]
Then
$
G(V_{\bm\alpha}D_{\bm v})^{-1}=GV_{\bm\alpha}^{-1}=(A|B)
$
where
\[
A =
\begin{pmatrix}
1&&& \\
&\ddots&& \\
&&1& \\
&&&0 \\
\end{pmatrix},~
B =
\begin{pmatrix}
\bm 0 \\
\vdots \\
\bm 0 \\
\bm b \\
\end{pmatrix}
\]
with $\bm b=(1,0,\dots,0)$, which shows that the first $k$  columns of $G(V_{\bm \alpha}D_{\bm v})^{-1}$ are not invertible.
\end{Example}

The following corollary gives a classification of part of our non-GRS MDS codes. Note that it is not clear that whether all constructed non-GRS MDS codes in Section 2 are GTRS codes or not.
\begin{Corollary}
The codes in Proposition \ref{pro34} with $\infty\notin\{\alpha_{1},\alpha_{2},\dots,\alpha_{n}\}$ are GTRS codes with the same $\bm\alpha$, $\bm v$.
\end{Corollary}
\begin{proof}
We compute
$$(G_{\bm\alpha,\bm v} + \beta E_{11})(V_{\bm\alpha}D_{\bm v})^{-1}=(I_{k}|\bm 0)+\begin{pmatrix}
v_{1}^{-1}\beta\bm m \\
\bm 0 \\
\vdots \\
\bm 0 \\
\end{pmatrix}$$
where $\bm m=(m_{0},m_{1},\dots,m_{n-1})$ is the first row of $V_{\bm\alpha}^{-1}$.

Since the inverse of a Vandermonde matrix is known, we have $$a\prod_{i=2}^{n}(x-\alpha_{i})=\sum_{j=0}^{n-1}m_{j}x^{j}$$ for some non-zero $a\in\mathbb{F}_{q}$.

If $\alpha_{i}=0$ for some $2\leq i$, then $m_{0}=0$ which implies the first $k$ columns of $$(I_{k}|\bm 0)+\begin{pmatrix}
v_{1}^{-1}\beta\bm m \\
\bm 0 \\
\vdots \\
\bm 0 \\
\end{pmatrix}$$
is an upper triangular matrix with non-zero diagonal entries.

If $\alpha_{i}\neq0$ for all $2\leq i$, then $m_{0}\neq0$ and $v_{1}^{-1}m_{0}\beta+1\neq0$ since the minimal polynomial of $\beta$ is of at least degree $2$ which also implies the first $k$ columns of $$(I_{k}|\bm 0)+\begin{pmatrix}
v_{1}^{-1}\beta\bm m \\
\bm 0 \\
\vdots \\
\bm 0 \\
\end{pmatrix}$$
is an upper triangular matrix with non-zero diagonal entries. By Proposition \ref{trs}, we complete the proof.
\end{proof}

\begin{Example}
Let $q=7^{2}$, $\mathbb{F}_{7^{2}}=\mathbb{F}_{7}(\theta)$ with $\theta^{2}+2=0$. Let $\bm\alpha=(1,\dots,6,0)$ and $\bm v=\bm 1$, then
$$
\left(
\begin{matrix}
1+\theta&1&1&1&1&1&1\\
1&2&3&4&5&6&0\\
1&4&2&2&4&1&0
\end{matrix}
\right)
$$
generate a $[7,3]$ non-GRS MDS code $C$ over $\mathbb{F}_{49}$ which is constructed from Proposition  \ref{pro34}.

Then $$(G_{\bm\alpha,\bm v} + \theta E_{11})V_{\bm\alpha}^{-1}=(I_{3}|\bm 0)+\begin{pmatrix}
v_{1}^{-1}\beta\bm m \\
\bm 0 \\
\vdots \\
\bm 0 \\
\end{pmatrix}=(I_{3}|\bm 0)+\begin{pmatrix}
\theta(0,6,6,6,6,6,6) \\
\bm 0 \\
\vdots \\
\bm 0 \\
\end{pmatrix}.$$
Set
$$
P=\begin{pmatrix}
1&\theta&\theta \\
0&1&0 \\
0&0&1 \\
\end{pmatrix},
$$
then
$$
P(G_{\bm\alpha,\bm v} + \theta E_{11})V_{\bm\alpha}^{-1}=\begin{pmatrix}
1&0&0&6\theta&6\theta&6\theta&6\theta \\
0&1&0&0&0&0&0 \\
0&0&1&0&0&0&0 \\
\end{pmatrix}.
$$
Then
$$
P(G_{\bm\alpha,\bm v} + \theta E_{11})V_{\bm\alpha}^{-1}V_{\bm\alpha}=\begin{pmatrix}
1&0&0&6\theta&6\theta&6\theta&6\theta \\
0&1&0&0&0&0&0 \\
0&0&1&0&0&0&0 \\
\end{pmatrix}V_{\bm\alpha}
$$
which implies the code $C$ is also a GTRS code
$GTRS_{7,3}[\bm\alpha,\bm t,\bm h,\bm\eta,\bm v]$ with $\bm v=\bm 1$, $\bm h=(0,0,0,0)$, $\bm t=(1,2,3,4)$ and $\bm\eta=(6\theta,6\theta,6\theta,6\theta)$.
\end{Example}


\begin{Lemma}\label{lem:grs}
Let $G$ be a generator matrix of a linear $[n,k]$ code $\mathcal{C}$. Then asking $G$ to be the generator matrix of a GTRS code is equivalent to ask for the existence of ${\bm \alpha}, {\bm v}$ such that the $k$ rows of
\[
G-
\begin{pmatrix}
\bm 1\\
\bm\alpha\\
\vdots \\
\bm \alpha^{k-1}
\end{pmatrix}D_{\bm v}
\]
are codewords of the GRS code
$GRS_{n,n-k}({\bm \alpha}, {\bm \alpha}^k*{\bm v})$
if all coefficients of $\bm\alpha$ are non-zero, and codewords of the code which extends $0$ in the $i$th coordinate to the GRS code $GRS_{n-1,n-k}({\bm \alpha}', ({\bm \alpha}')^k*{\bm v}')$ for ${\bm \alpha}'=(\alpha_1,\ldots,\widehat{\alpha_{i}},\ldots,\alpha_{n})$ and ${\bm v}'=(v_1,\ldots,\widehat{v_{i}},\ldots,v_{n})$ if $\alpha_{i}=0$ where ~
 $\widehat{}$~ means deleting the corresponding entry in the vector.
\end{Lemma}
\begin{proof}
By (\ref{eq:gtrs}), we know
\[
G-
\begin{pmatrix}
\bm 1\\
\bm\alpha\\
\vdots \\
\bm \alpha^{k-1}
\end{pmatrix}D_{\bm v}=M\left(
\begin{matrix}
\bm\alpha^{k}\\
\bm\alpha^{k+1}\\
\vdots\\
\bm\alpha^{n-1}
\end{matrix}
\right)D_{\bm v}=M\begin{pmatrix}
\bm 1\\
\bm\alpha\\
\vdots \\
\bm\alpha^{n-k-1}
\end{pmatrix}
D_{{\bm \alpha}^k}
D_{\bm v}.
\]
Then we get the result immediately.

\end{proof}

\begin{Proposition}
Let $G$ be a generator matrix of a linear $[n,k]$ code $\mathcal{C}$. Then asking $G$ to be the generator matrix of a GTRS code is equivalent to ask for the existence of ${\bm \alpha}, {\bm v}$ such that $G$ satisfy
\[
\begin{pmatrix}
\bm 1\\
\bm\alpha\\
\vdots \\
\bm \alpha^{k-1}
\end{pmatrix}D_{\bm w}
G^T =
D'_{{\bm u}}
\]
where $D'_{{\bm u}}$ is an anti-diagonal matrix with non-zero coefficients ${\bm u}$ on the anti-diagonal. And ${\bm w}$ is any non-zero codeword in the 1-dimensional code which is the dual of $GRS_{n,n-1}(\bm\alpha,\bm\alpha^k*\bm v)$ when all coefficients of $\bm\alpha$ are non-zero, $\bm w$ is from inserting a $0$ in the $i$th coordinate of ${\bm w^{'}}$ where ${\bm w^{'}}$ is any non-zero codeword in the 1-dimensional code which is the dual of $GRS_{n-1,n-2}(\bm\alpha^{'},(\bm\alpha^{'})^k*\bm v^{'})$ for  ${\bm \alpha}'=(\alpha_1,\ldots,\widehat{\alpha_{i}},\ldots,\alpha_{n})$ and ${\bm v}'=(v_1,\ldots,\widehat{v_{i}},\ldots,v_{n})$ if $\alpha_{i}=0$.
\end{Proposition}
\begin{proof}
We first suppose all coefficients of $\bm\alpha$ are nonzero, and it is same for other case.
By Lemma \ref{lem:grs}, the rows of \[
G-
\begin{pmatrix}
\bm 1\\
\bm\alpha\\
\vdots \\
\bm \alpha^{k-1}
\end{pmatrix}D_{\bm v}
\] must be codewords of $GRS_{n,n-k}({\bm \alpha},{\bm \alpha}^k*v)$, which means they must live in the kernel of the corresponding parity check matrix.

A parity check matrix $H$ for this code is a generator matrix of its dual, which is known to be $GRS_{n,k}({\bm \alpha, {\bm w}})$ for ${\bm w}$  any non-zero codeword in the
$1$-dimensional code which is the dual of the code $GRS_{n,n-1} ({\bm \alpha}, {\bm \alpha}^k*{\bm v})$; in other words, it satisfies
$$
\sum_{l=0}^{n-1}w_l \alpha_l^kv_l h(\alpha_l) = 0
$$
for any polynomial $h$ of degree $n-1$.
Explicitly, we have
\[
H=
\begin{pmatrix}
\bm 1\\
\bm\alpha\\
\vdots \\
\bm \alpha^{k-1}
\end{pmatrix}D_{\bm w}.
\]

A general row from the left hand-side is
\[
{\bm g}-{\bm \alpha}^iD_{\bm v},~0\leq i \leq k-1
\]
so for such a row to be a codeword, it is enough to apply $H$ and make sure the result is zero:
\begin{eqnarray*}
H({\bm g}^T-D_{\bm v}^T ({\bm \alpha}^i)^T)
& = &
\begin{pmatrix}
\bm 1\\
\bm\alpha\\
\vdots \\
\bm \alpha^{k-1}
\end{pmatrix}D_{\bm w}{\bm g}^T -
\begin{pmatrix}
\bm 1\\
\bm\alpha\\
\vdots \\
\bm \alpha^{k-1}
\end{pmatrix}D_{\bm w * \bm v}({\bm \alpha}^i)^T.
\end{eqnarray*}
Next ${\bm \alpha}^j({\bm \alpha^i})^T= \sum_l \alpha_l^{j+i}$ for $0\leq i,j \leq k-1$ and $h_j(\alpha_l)=\alpha_l^j$,  for $h_j(x)=x^j$, $0\leq j \leq n-1$, so
\begin{eqnarray*}
\begin{pmatrix}
\bm 1\\
\bm\alpha\\
\vdots \\
\bm \alpha^{k-1}
\end{pmatrix}D_{\bm w * \bm v}({\bm \alpha}^i)^T
&=&
\begin{pmatrix}
\sum_l w_lv_l \alpha_l^{i} \\
\sum_l w_lv_l\alpha_l^{1+i} \\
\vdots \\
\sum_l w_lv_l\alpha_l^{k-1+i}
\end{pmatrix} \\
&=&
\begin{pmatrix}
\sum_l w_lv_l \alpha_l^kh_{i-k}(\alpha_l) \\
\sum_l w_lv_l\alpha_l^kh_{1+i-k}(\alpha_l) \\
\vdots \\
\sum_l w_lv_l\alpha_l^kh_{-1+i}(\alpha_l)
\end{pmatrix}
\end{eqnarray*}
where all rows such that $h_j(x)$ has an index $0\leq j\leq n-1$  are zero  by definition of ${\bm w}$, namely
$
\sum_{l=0}^{n-1}w_l \alpha_l^kv_l h(\alpha_l) = 0.
$
When $i=0$, then all rows are zero but for the last one. In general, all rows are zero but for row $k-i$. This results in an anti-diagonal matrix $D'_{\bm u}$.
\end{proof}

\section{Conclusion}
Motivated by the classification of MDS codes,
we proposed a generic construction of MDS codes, out of which we exhibited families of non-GRS MDS codes. We discussed connections with GTRS codes, but we were not yet able to tell whether some of the proposed families might contain non-GTRS codes, since the characterization of (non-)GTRS codes remains an open but seemingly challenging question.

\vskip 4mm

\noindent {\bf Acknowledgement.} This work was supported by NSFC (Grant Nos. 11871025, 12271199) and China Scholarship Council (No. 202306770055). The hospitality of the division of mathematical sciences at Nanyang Technolgical University during the first author's visit is gratefully acknowledged.


\end{document}